\documentclass[11pt,a4paper]{article}
\usepackage{fullpage}
\usepackage{amssymb,amsmath,amsthm}
\usepackage{tikz}
\usepackage{subfigure}
\usetikzlibrary{automata,positioning}

\theoremstyle{plain}
\newtheorem{theorem}{Theorem}
\newtheorem{lemma}{Lemma}
\newtheorem{prop}{Proposition}
\newtheorem{corollary}{Corollary}
\newtheorem{fact}{Fact}
\theoremstyle{definition}
\newtheorem{definition}{Definition}

\newcommand{\bfv}[1]{\mathbf{#1} } 

\begin{document}

\title{On Restricting No-Junta Boolean Function and Degree Lower Bounds by Polynomial Method}

\author{
Chia-Jung Lee\thanks{Department of Computer Science, National Chiao Tung University, Hsinchu, Taiwan. }
\and
Satya V. Lokam\thanks{Microsoft Research India, Bangalore, India. }
\and
Shi-Chun Tsai \footnotemark[1] 
\and 
Ming-Chuan Yang \footnotemark[1]
}

\maketitle

\begin{abstract}

Let $\mathcal{F}_{n}^*$ be the set of Boolean functions depending on all $n$ variables.
We prove that for any $f\in \mathcal{F}_{n}^*$,
$f|_{x_i=0}$ or $f|_{x_i=1}$ depends on the remaining $n-1$ variables, for some variable $x_i$.
This existent result suggests a possible way to deal with general Boolean functions via its subfunctions of some restrictions. 

As an application, we consider the degree lower bound of representing polynomials over finite rings. Let $f\in \mathcal{F}_{n}^*$ and 
denote the exact representing degree over the ring  $\mathbb{Z}_m$ (with the integer $m>2$) as $d_m(f)$.  
Let $m=\Pi_{i=1}^{r}p_i^{e_i}$, where $p_i$'s are distinct primes, and $r$ and $e_i$'s are positive integers. If $f$ is symmetric, then $m\cdot d_{p_1^{e_1}}(f)\cdots d_{p_r^{e_r}}(f) > n$. If $f$ is non-symmetric, by the second moment method we prove almost always
$m\cdot d_{p_1^{e_1}}(f)\cdots d_{p_r^{e_r}}(f) > \lg{n}-1$. In particular, as $m=pq$ where $p$ and $q$ are arbitrary distinct primes, we have $d_p(f)d_q(f)=\Omega(n)$ for symmetric $f$ and $d_p(f)d_q(f)=\Omega(\lg{n}-1)$ almost always for non-symmetric $f$. Hence any $n$-variate symmetric Boolean function can have exact representing degree $o(\sqrt{n})$ in at most one finite field, and for non-symmetric functions, with $o(\sqrt{\lg{n}})$-degree in at most one finite field.

\end{abstract}


\section{Introduction}
The random restriction method and the polynomial method are powerful tools in computational complexity \cite{Jun12}; the former is applied to make a function become easier subfunctions (such as the classical switching lemma \cite{Has86}) and the polynomials are good computation or approximation models of Boolean functions \cite{Am06,BL14,Raz87}.
A subfunction of a given Boolean function is obtained by fixing some variables $0/1$ assignments (a.k.a. a restriction). 
In is work we propose an idea of dealing with general Boolean functions by finding a well-studied while non-trivial substructure (for instance, symmetric or monotone subfunctions). To do this, we need the subfunctions maintain some property after fixing some variables.   

In the literature (such as \cite{Has86,Jun12,Tal14}), arguments via  random or adaptive restrictions on a Boolean function are targeted to prove a very low probability of sustaining hard subfunctions after some restrictions. 
That is, a Boolean function is very likely to degenerate after restrictions. 
For example, if any variable is assigned as $1$ then the $\mathsf{OR}_n$ function becomes a constant. It is not clear that for arbitrary Boolean function depending on all variables, does there exist a subfunction that depends on the rest variables (the so-called no-junta subfunction)? This is important as most of complexity measures are conditioned on the number of influential variables. We give a positive answer in this work.

As an application of the existence of no-junta subfunction, we consider the degree lower bound of general Boolean functions. The degree of a polynomial exactly representing a boolean function relates to many classical complexity measures \cite{Jun12}, including 
the decision tree complexity, the circuit complexity \cite{Raz87} and the quantum complexity \cite{Am06}, etc. It is easy to acquire the representing degree for a specific function since the polynomial can be constructed by the interpolation based on the truth table. However, this cannot give a nontrivial degree lower bound for all functions. Nisan and Szegedy \cite{NS94} first proved a lower bound for real value representing polynomials. When the polynomial coefficients are in finite fields, a breakthrough was made by Gopalan et al. \cite{GLS10}. 
In this work we give lower bounds with simpler forms and different proofs. 


The rest of this paper is organized as follows. In Section 2, we introduce some notations and definitions. In Section 3 we prove any $n$-variate function has a subfunction depending on the rest of $n-1$ variables. Section 4 shows the degree bounds. Section 5 concludes the paper.

\section{Preliminaries}
Define $\mathcal{F}_n=\left\{ f:\{0,1\}^{n} \to \{0,1\} \right\}$, which is the collection of all
$n$-variable Boolean functions. 
Let $[n]=\{1,\cdots,n\}$, $\bfv{x}=(x_1,\cdots,x_n)\in \{0,1\}^n$. Note that $\mathcal{F}_n$ contains juntas. A $k$-junta is a function $f(\bfv{x})$ represented as a formula of $n$ variables while in fact there exists a proper subset $K\subset [n]$ of size $k<n$ that decides the value of $f$, and all variables in $[n]\setminus K$ have no influence on the output of $f$. Denote $\mathcal{F}_{n}^*$ as the subset of $\mathcal{F}_n$ that excludes all juntas.
Hence,  if $f\in \mathcal{F}_{n}^*$, then $f$ depends on all $n$ variables. 

Consider an index subset $I\subseteq[n]$ and a partial assignment $\rho_{I}$ on the variables indexed by $I$. For a given $\rho_I$, we call a variable in $I$ the assigned variable and a variable in $I^c$ the free variable. 
We define the subfunction $f(\bfv{x})|_{\rho_I}$ to be the function derived by restricting the variables according to $\rho_I$ and keeping the variables in $I^c$ free. Take the address function \cite{Jun12} 
as an example: $\mathsf{Address}(x_1,x_2,x_3,x_4,x_5,x_6)$ outputs the value of $x_{z}$ where $z={x_1}\cdot 2^0+{x_2}\cdot 2^1+3$. Let $\rho_{\{1,6\}}$ be a restriction such that $x_1=1, x_6=0$ then $\mathsf{Address}(x_1,x_2,x_3,x_4,x_5,x_6)|_{\rho_{ \{1,6\} }}=x_4-x_2x_4$.
Note this subfunction is a 2-junta not depending on the free variables $x_3$ and $x_5$.
As mentioned, for an arbitrary $n$-variate function the existence of its no-junta subfunction was not clear. 

Let $\mathbb{Z}_m$ be the ring of $\{0,\cdots,m-1\}$, where the integer $m\geq 2$.
Let $\mathbb{Z}_m[\bfv{x}]$ be the set of all polynomials over $\mathbb{Z}_m$, and similarly for $\mathbb{R}[\bfv{x}]$ and $\mathbb{Z}[\bfv{x}]$.
For $f \in \mathcal{F}_n$, let $\left\{ f(D):D\subseteq [n] \right\}$ be its truth table. 
Then $f$ can be represented by $F(\bfv{x})=\sum_{D\subset [n]}{f(D)\Pi_{i\in D}x_i\Pi_{i\not\in D}(1-x_i)}\in \mathbb{Z}[\bfv{x}]$ (also in $\mathbb{R}[\bfv{x}]$). This implies the existence of the exact representation and the uniqueness is easy to prove \cite{GLS10,Jun12}. 
Furthermore, since all coefficients in the expansion of $F(\bfv{x})$ are integers and $x_i$'s are in $\{0,1\}$, $F(\bfv{x})$ is a  multilinear polynomial. Also observe that $F(\bfv{x})\pmod{m}\in\mathbb{Z}_m[\bfv{x}]$. Denote $F(\bfv{x})\pmod{m}$ by $P(\bfv{x})$. Obviously we have $\deg(F(\bfv{x}))\geq \deg(P(\bfv{x}))$. That is, for a given $f$, $d_{\mathbb{Z}}(f)\geq d_{m}(f)$. 
For $f\in \mathcal{F}_{n}^*$,
the Nisan-Szegedy bound \cite{NS94} states that $d_{\mathbb{R}}(f)\geq \log{n}-O(\log{\log{n}})$.
Observe that $\mathsf{PARITY}(\bfv{x})$ over $\mathbb{Z}_2$ can be computed by $\sum_{i=1}^nx_i\pmod{2}$, i.e. $d_{2}(\mathsf{PARITY})=1$, which shows the difference between $d_{\mathbb{R}}(f)$ and $d_{m}(f)$.
Also note that Gopalan et al. \cite{GLS10} proved $ {\lceil \lg p\rceil}\cdot p^{2 d_{p}(f)} \cdot d_{p}(f)\cdot d_{q}(f) \geq n$.

Suppose the total degree of $P(\bfv{x})$ is $d$, and write $P(x_1,\cdots,x_n)=\sum_{D\subseteq [n];|D|\leq d}{c_D\Pi_{j\in D}x_j}\pmod{m}$, where $c_D\in \mathbb{Z}_m$ for all $D\subseteq [n]$. 
Sometimes we treat $D\subseteq [n]$ as a variable, such as $f(D)$ or $P(D)$; by this we mean $f(D)=f(x_1,\cdots,x_n)$ where $x_i=1\Leftrightarrow i\in D$. 
Denote $|\bfv{x}|=\sum_{i=1}^{n}x_i$ while $|D|$ is the cardinality of $D\subseteq [n]$. If $f$ is symmetric then the output value of $f$ is decided by $|\bfv{x}|$. 
The symmetry makes all monomials of the same degree have the same coefficient, that is, for all $D$ and $D'$ with $|D|=|D'|=k$, we have $c_D=c_{D'}=c_k$. It is easy to prove 
\begin{fact} \label{sympoly}
If $P$ is a symmetric polynomial of degree $d$, then for any $A\subseteq [n]$,
$$
P(A)=\sum_{k=0}^{d}{c_k}\binom{|A|}{k}
$$
\end{fact}

The following property is helpful to expose the relation between the binomial coefficients and the degree of a polynomial.

\begin{lemma}(See \cite{Tsai96}.)\label{cycle}
Let $m=\Pi_{i=1}^{r}{p_i^{e_i}}$ be a positive integer, where $p_i$'s are distinct primes and $e_i$'s are positive integers.
Let $L_k=\Pi_{i=1}^{r}{p_i^{e_i+\lfloor \log_{p_i}{k} \rfloor}}$. Then for any given nonnegative integers $s$, $j$ and $k$, we have 
$$
\binom{sL_k+j}{k}\equiv \binom{j}{k}\pmod{m}.
$$

\end{lemma}
For a given symmetric polynomial of degree $d$, we will consider $k=0,1,\cdots,d$. 
For saving space, we  abbreviate $a\equiv b \pmod{m}$ as $a\equiv_{m} b$.
Observe that if $k\leq d$ then $L_k|L_d$. This means $\forall k\leq d$, $\binom{sL_d+j}{k}\equiv_{m} \binom{j}{k}$.
Besides, note that $\binom{j}{k}=0$ if $j<k$. 

\section{Existence of No-Junta Subfunction}
Let 
$f\in \mathcal{F}_n$
be a Boolean function on $n$ variables $x_1, \cdots, x_n$. We say $f$ depends
on the $i$-th variable $x_i$ if there is some input $a$ such that $f(a)\neq f(a^{(i)})$, where $a^{(i)}$
is obtained from $a$ by flipping the value of its $i$-th coordinate. The function $f$ is called \textit{nondegenerate} if $f$ depends on all its variables(, i.e. $f\in \mathcal{F}_n^{*}$).

The variable $x_i$ is said to be \textit{useful} for $f$ if $f$ depends on $x_i$; otherwise $x_i$ is \textit{useless}  for $f$. We will denote the set of (indices of) useless variables for f by $U(f) := \{i : f \mbox{ does not depend on } x_i \}$.
Note that the notion of a useless variable presupposes (often implicitly) a universe of variables
on which f is defined.

For $i\in [n]$ and $b\in\{0, 1\}$, 
the restriction $ f|_{x_i =b} : \{0, 1\}^{[n]\setminus\{i\}}\to\{0, 1\}$
is defined as the subfunction on variables $x_j$, $j\neq i$, obtained by fixing $x_i = b$. For notational convenience, we
denote $f|_{x_i =b}$ by $f_{ib}$.

Our main theorem is

\begin{theorem} \label{no_junta}
If $f\in \mathcal{F}_n^{*}$ then there exists an $i\in [n]$ such that at least one of $f_{i0}$ or $f_{i1}$ is nondegenerate, i.e., it depends on all the variables $[n]\setminus \{i\}$.
\end{theorem}
 
We start with an obvious observation:
\begin{prop} \label{juntabasic}
(i) If $x_i$ is useless for $f$, it is useless for any restriction of $f$. In particular,
$i\in U(f)$, then for all $j\in [n]\setminus\{i\}$, $i\in  U(f_{j0})\cap U(f_{j1})$.\\
(ii) On the other hand, if $i$ is useless for both $f_{j0}$ and $f_{j1}$ for some $i$ and $j$, then $i$ is useless
for $f$ as well. In notation, $\exists j$ such that $i\in U(f_{j0})\cap U(f_{j1})\Rightarrow i\in U(f)$.
\end{prop}
\begin{proof}
For $i\neq j\in [n]$, the conclusions are clear by observing $f=(1-x_j)\cdot f_{j0}+x_j\cdot f_{j1}$. 
\end{proof}

Given a boolean function $f\in \mathcal{F}_n$, we construct a digraph $G_f=(V,E)$ with $V=[n]$ and a directed edge ${i}\overset{b}{\longrightarrow} j$ labeled by $b$ if and only if $j\in U(f_{ib})$, i.e. 
 $E=\{i\overset{b}{\longrightarrow} j: i\in [n], j\in U(f_{ib}),b=0\mbox{ or }1\}$. The digraph $G_f$ has the following properties:

\begin{prop} \label{trans_cycle}
(\textbf{Transitivity}:) 
Suppose in $G_f$ we have a path $i\overset{b_i}{\longrightarrow} j \overset{b_j}{\longrightarrow}k$ for distinct $i,j,k\in [n]$. Then the edge $i\overset{b_i}{\longrightarrow} k$ is also in $G_f$, i.e., $k\in U(f_{ib_i})$. More generally, if a path of distinct variables $i\overset{b_i}{\longrightarrow}j_1\overset{b_{j_1}}{\longrightarrow}\cdots \overset{b_{j_{t-1}}}{\longrightarrow} j_t$ exists in $G_f$, then $\{{j_1},\cdots,{j_t}\}\subseteq U(f|_{x_i=b_0})$.\\ 
(\textbf{Cycles}:) 
If there are two cycles through $i$ containing edges $i\overset{0}{\longrightarrow} j$ and $i\overset{1}{\longrightarrow} k$, all three variables $x_i, x_j, x_k$ are useless for $f$.
\end{prop}

Observe that if $G_f$ has a cycle, say ${5}\overset{0}{\longrightarrow} {2}\overset{1}{\longrightarrow} {4}\overset{0}{\longrightarrow}5$, then we can apply Transitivity to reduce it to a collection of cycles of two vertices, e.g., ${5}\overset{0}{\longrightarrow}{4}\overset{0}{\longrightarrow}5$, etc. 

\begin{proof} (Transitivity:)
Note that $i,j,k$ are distinct, otherwise it is in fact a cycle. 
In the path, by definition of $j \overset{b_j}{\longrightarrow}k$ we have $f|_{x_j=b_j,x_k=0}=f|_{x_j=b_j,x_k=1}$. 
Observe that
$$\aligned
& f|_{x_{j}=b_j,x_{k}=0}= f|_{x_{j}=b_j,x_{k}=1} \\
& \Rightarrow 
f|_{x_i=b_i,x_{j}=b_j,x_{k}=0}= f|_{x_i=b_i,x_{j}=b_j,x_{k}=1}
\endaligned
$$
Besides, $i \overset{b_i}{\longrightarrow}j$ implies $f|_{x_i=b_i,x_j=b_j}=f|_{x_i=b_i,x_j=1-b_j}$. 
Hence,
$$\aligned
& f|_{x_i=b_i,x_j=b_j,x_k=0}= f|_{x_i=b_i,x_j=1-b_j,x_k=0} \\
& =
f|_{x_i=b_i,x_j=b_j,x_k=1}= f|_{x_i=b_i,x_j=1-b_j,x_k=1}
\endaligned
$$
That is, $k\in U(f_{i b_i})$.
 
For the general case of  $i\overset{b_i}{\longrightarrow}j_1\overset{b_{j_1}}{\longrightarrow}\cdots \overset{b_{j_{t-1}}}{\longrightarrow} j_t$, it is easy to prove by induction on $t$ with almost the same argument as above.

(Cycles:) 
First observe that for a given cycle $i\overset{b_i}{\longrightarrow}j\overset{b_j}{\longrightarrow}i $, by definition we have $f|_{x_i=b_i,x_j=b_j}=f|_{x_i=b_i,x_j=1-b_j}$ because of $i\overset{b_i}{\longrightarrow}j$, and $f|_{x_i=b_i,x_j=b_j}=f|_{x_i=1-b_i,x_j=b_j}$ because of $j\overset{b_j}{\longrightarrow}i$. Combining these two equations, it is clear to have $f|_{x_i=b_i,x_j=b_j}=f|_{x_i=b_i,x_j=1-b_j}=f|_{x_i=1-b_i,x_j=b_j}$. Note that there is no requirement on $f|_{x_i=1-b_i,x_j=1-b_j}$.

Now, for the given pair of cycles $i\overset{0}{\longrightarrow}j\overset{b_j}{\longrightarrow}i$ and $i\overset{1}{\longrightarrow}k\overset{b_k}{\longrightarrow}i$ in $G_f$, we can list the corresponding equations: 

$$
\left\{
\aligned
& f|_{x_i=0,x_j=b_j,x_k=b_k}=f|_{x_i=0,x_j=1-b_j,x_k=b_k}=f|_{x_i=1,x_j=b_j,x_k=b_k} \\ 
& f|_{x_i=0,x_j=b_j,x_k=1-b_k}=f|_{x_i=0,x_j=1-b_j,x_k=1-b_k}=f|_{x_i=1,x_j=b_j,x_k=1-b_k} \\ 
& f|_{x_i=1,x_j=b_j,x_k=b_k}=f|_{x_i=1,x_j=b_j,x_k=1-b_k}=f|_{x_i=0,x_j=b_j,x_k=b_k} \\
& f|_{x_i=1,x_j=1-b_j,x_k=b_k}=f|_{x_i=1,x_j=1-b_j,x_k=1-b_k}=f|_{x_i=0,x_j=1-b_j,x_k=b_k} \\ 
\endaligned
\right.
$$

Note that the above equations are related by $f|_{x_i=0,x_j=b_j,x_k=b_k}$ in the first and third equations, $f|_{x_i=0,x_j=1-b_j,x_k=b_k}$ in the first and fourth, and $f|_{x_i=1,x_j=b_j,x_k=1-b_k}$ in the second and third. 
Besides, they contain all eight configurations. 
This shows that $f$ is independent of $x_i, x_j$ and $x_k$.
I.e., $f$ is not sensitive on variables $x_i,x_j$ and $x_k$.
\end{proof}

\begin{definition}
For $f\in \mathcal{F}_n$ and $S\subseteq [n]$, $U1(f,S):=\bigcup_{i\in S}(U(f_{i0})\cup U(f_{i1}))$. In words, a variable $j$ is in $U1(f,S)$ if it is useless for a restriction of $f$ that fixes one variable to some value.
\end{definition}

\begin{prop} \label{closure}
(\textbf{Closure Property}:) For any $i\in[n]$ and $b\in \{0,1\}$, we have 
$$
U1(f,U(f_{ib}))\subseteq U(f_{ib})\cup \{i\}.
$$ 
Furthermore, if for every $k\in U(f_{ib})$ and every $b'\in \{0,1\}$, $i\not\in U(f_{kb'})$, then 
$$
U1(f,U(f_{ib}))\subseteq U(f_{ib}).
$$
\end{prop}

\begin{proof}
For any $j\in U(f_{ib_i})$ and $j'\in U(f_{jb_j})\setminus \{i\}$ (where $b_i,b_j\in\{0,1\}$) we can append ${j}\overset{b_j}{\longrightarrow} j'$ to $i\overset{b_i}{\longrightarrow} j$ and obtain a path. 
If for some $j\in U(f_{ib_i})$, $i\in U(f_{j0})\cup U(f_{j1})$, then  we have $U(f_{j0})\cup U(f_{j1})\subseteq U(f_{ib_i})\cup \{i\}$. Hence,
$$
U1(f;U(f_{ib_i}))
=\bigcup_{j\in U(f_{ib_i}),b=0,1}{U(f_{jb})}
\subseteq U(f_{ib_0})\cup \{i\}.
$$

On the other hand, if $\forall j\in U(f_{ib_i}), i\not\in U(f|_{j0})\cup U(f|_{j1})$, then $U(f_{j0})\cup U(f_{j1})\subseteq U(f_{ib_i})$. This concludes that
$U1(f,U(f_{ib}))\subseteq U(f_{ib})$. 

\end{proof}

\begin{proof} (\textbf{of Theorem \ref{no_junta}}) 
We prove the theorem by contraction. So, assume $f$ is nondegenerate but every restriction $f_{ib}$, $i\in [n]$ is degenerate.  
Then we have $U(f|_{x_i=0}) \neq \emptyset$ and $U(f|_{x_i=1})\neq \emptyset$ for all $i\in [n]$. 

Thus, we can assume that $f$ satisfies the following property:\\
$$
(*)\mbox{  } \forall i\in [n], U(f|_{x_i=0}) \neq \emptyset \mbox{ and } U(f|_{x_i=1})\neq \emptyset.
$$

We prove the claim below:\\
\textbf{Main Claim:} {If $f$ has property $(*)$ and $U(f;S)\subseteq S$ for some $S\subseteq [n]$ with $|S|\geq 2$, then $\exists j\in S$ such that $f|_{j0}= f|_{j1}$, i.e., $j$ is useless for $f$.}

Since $U1(f,[n])\subseteq [n]$ and $n\geq 2$, the claim can be applied with $S=[n]$ and it implies $f$ is degenerate. This is a contraction.

\end{proof}

\begin{proof} (\textbf{of Main Claim}) We prove this by induction on $|S|$.

Base case ($|S|=2$): 
Note that $|S|=2$ with $U(f;S)\subseteq S$ implies $U(f|_{i0})=\{j\}$ and $U(f|_{i1})=\{j\}$. Since $j\in U(f_{i0})\cap U(f_{i1})$, Proposition \ref{juntabasic}(ii) shows $j\in S$ is useless for $f$.

Inductive step:
Assume $|S|=s\geq 3$ and the claim holds for all $S$ with $|S|\leq s-1$.

Consider an arbitrary vertex $j^{*}\in S$. By $(*)$ there is at least one $0$-edge $j^{*}\overset{0}{\longrightarrow} u$ and at least one $1$-edge $j^{*}\overset{1}{\longrightarrow} v$. As shown in Figure \ref{case1}, if $u=v$ then $u\in U(f_{j^{*}0})\cap U(f_{j^{*}1})$ and by Proposition \ref{juntabasic}(ii), $u$ is useless and we are done. So, we may assume $u\neq v$. If there is a directed path from $u$ back to $j^{*}$ and a directed path from $v$ back to $j^{*}$ (as shown in Figure \ref{case2}), then $j^{*},u,v$ satisfy the Cycles property of Proposition \ref{trans_cycle} and all of them are useless for $f$ and we are again done.  

Hence, we can assume that there is no cycle \textit{via}, w.l.o.g., any $1$-edge leaving $j^{*}$ back to $j^{*}$. This case is shown in Figure \ref{case3}, where we use the dashed edge to emphasize $G_f$ does not have such edge. Let $S':=U(f_{j^{*}1})$. By the foregoing assumption, no $k\in S'$ can have an edge going back to $j^{*}$ and hence $j^{*}\not\in S'$. By the Closure property (Proposition \ref{closure}, the "furthermore" part), we have $U1(f,S')\subseteq S'$. We also have that $|S'|\leq s-1$. Moreover, $|S'|\geq 2$ since there is at lest one $1$-edge $j^{*}\overset{1}{\longrightarrow} v$ and by $(*)$, there are two outgoing edges out of $v$ neither of which can go back to $j^{*}$ (the heads of those two edges could coincide $--$ however, that'd immediately imply the claim). Hence we can apply the induction hypothesis to $S'$ and conclude there is a vertex $w\in S'$ that is useless for $f$. Note that $S'\subset S$ by the Transitivity property since $j^{*}\in S$ and $U1(f,S)\subseteq S$. Hence $w\in S$ and is useless for $f$.   
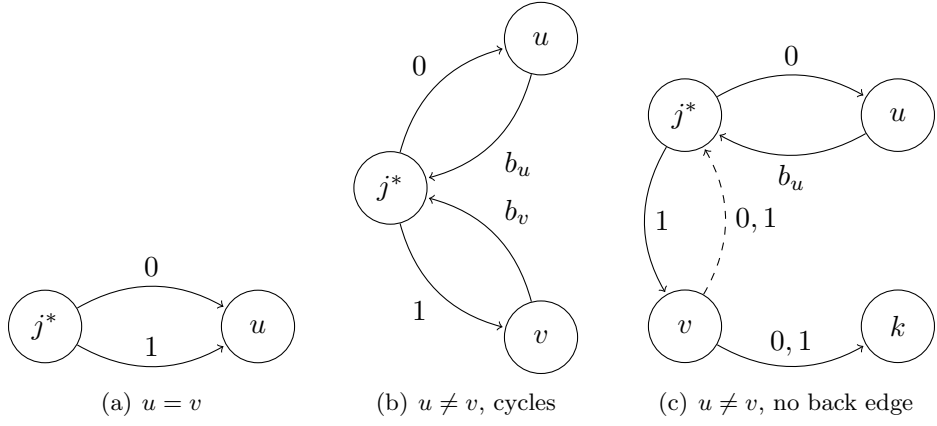
\begin{figure}
\centering
\subfigure[$u=v$]{
\begin{tikzpicture}[shorten >=1pt,node distance=2.8cm,on grid,auto] 
   \node[state] (j)  {$j^*$};
   \node[state] (u) [right=of j]{$u$};  
   \path[->] 
    (j) edge  [bend left] node  {$0$} (u)
        edge  [bend right] node {$1$} (u);
\end{tikzpicture}
\label{case1}}
\quad
\subfigure[$u\neq v$, cycles]{
\begin{tikzpicture}[shorten >=1pt,node distance=2.8cm,on grid,auto] 
   \node[state] (j)  {$j^*$};
   \node[state] (u) [above right =of j]{$u$};  
   \node[state] (v) [below right =of j]{$v$};  
   \path[->] 
    (j) edge  [bend left] node  {$0$} (u)
        edge  [bend right] node [swap] {$1$} (v)
    (u) edge  [bend left] node  {$b_u$} (j)
    (v) edge  [bend right] node [swap] {$b_v$} (j)   ;
\end{tikzpicture}
\label{case2}}
\quad
\subfigure[$u\neq v$, no back edge]{
\begin{tikzpicture}[shorten >=1pt,node distance=2.8cm,on grid,auto] 
   \node[state] (j)  {$j^*$};
   \node[state] (u) [right=of j]{$u$};  
   \node[state] (v) [below=of j]{$v$}; 
   \node[state] (k) [right=of v]{$k$};       
   \path[->] 
    (j) edge  [bend left] node  {$0$} (u)
        edge  [bend right] node {$1$} (v)
    (u) edge  [bend left] node  {$b_u$} (j)
    (v) edge  [bend right] node {$0,1$} (k);
   \path[->, dashed]    
    (v) edge [bend right] node [swap]{$0,1$} (j);
\end{tikzpicture}
\label{case3}}
\caption{ Graphs for Inductive step} 
\end{figure}  

\end{proof}

Many complexity measures are hard to be obtained for general functions but easier for that with some property (such as monotone or symmetric). This suggests we can get some complexity bounds of $f$ through its subfunction $f|_{\rho_I}$.
However, as the example of the $\mathsf{Address}$ function illustrates, even $f\in \mathcal{F}_{n}^*$, usually 
$f|_{\rho_I}\not\in \mathcal{F}_{n-|I|}^*$. 
The implication of Theorem \ref{no_junta} promises that for given $0<k<n$ there is one restriction ${\rho_I}$ with $n-|I|=k$ such that $f|_{\rho_I}$ is no-junta. 
An crucial observation is that any no-junta and symmmetric $f|_{\rho_I}\in \mathcal{F}_{n-|I|}^*$; this motivates us to derive a degree lower bound for non-symmetric Boolean functions from symmetric subfunctions.

\section{Degree Bounds} 

\begin{theorem}
Let $f\in \mathcal{F}_{n}^*$ be symmetric, $p_i$'s be distinct primes, and $r$ and $e_i$'s be positive integers. 
Let $m=\Pi_{i=1}^{r}p_i^{e_i}$. Then
$$
m\cdot d_{p_1^{e_1}}(f)\cdots d_{p_r^{e_r}}(f) > n.
$$
\end{theorem}
The idea of the proof is to find a particular subset $A$ such that the output values of representing polynomials on $A$ are distinct over two $\mathbb{Z}_{p_i^{e_i}}$'s, which is a contradiction; hence this $A$ must not exist.   
\begin{proof}
We only have to consider $f\in \mathcal{F}_{n}^*$ satisfying $f(\emptyset)=0$ since $\forall i$, $d_{p_i^{e_i}}(f)=d_{p_i^{e_i}}(\neg f)$. 
Besides, $f$ is no-junta and symmetric means there exists a weight $\tau\in [n]$ s.t. $\forall D\subseteq [n]$ with $|D|=\tau \leq d_{min}$, $f(D)=1$, where $d_{min}=\min_{i\in [r]}(d_{p_i^{e_i}}(f))$. W.L.O.G. let $d_{min}=d_1$. 
Let $P_{i}(\bfv{x})\in \mathbb{Z}_{p_i^{e_i}}[\bfv{x}]$ s.t. $P_{i}(\bfv{x})\pmod{p_i^{e_i}} = f(\bfv{x})$ for all $\bfv{x}\in \{0,1\}^n$. Denote $d_i=\deg(P_i(\bfv{x}))$.
Since $f$ is symmetric, each $P_i(A)$ can be written as $\sum_{k=0}^{d_i}{c_{P_i,k}\binom{|A|}{k}}$ by Fact \ref{sympoly}. 
Define $L_{i}=p_i^{e_i+\lfloor \log_{p_i}{d_i} \rfloor}$.
Then by Lemma \ref{cycle} for any nonnegative integers $s$, $j$ and $k$, we have 
$$
\forall k\leq d_{min}, \binom{sL_{i}+j}{k}  \equiv_{p_i^{e_i}} \binom{j}{k}.
$$
Consider $A\subseteq [n]$ such that 
$|A|\equiv \tau \pmod{L_{1}}$ and for $i\geq 2$, $|A|\equiv 0 \pmod{L_{i}}$. 
Let $L:=\Pi_{i=1}^r {L_i}$ and $L_i^{'}=L/L_i$. 
Then by the CRT, the unique solution (in the sense of modular convergence) is
$$
\aligned
|A|=
& \sum_{i=2}^r 0\cdot L_i^{'}\cdot [{L_i^{'}}^{-1}\pmod{L_i}]\\
   & +\tau\cdot L_1^{'}\cdot [{L_1^{'}}^{-1}\pmod{L_1}] \pmod{L}
    \leq L
\endaligned
$$
For this $|A|$ and by Fact \ref{sympoly}, 
$$\aligned
P_1(A)&=\sum_{k=0}^{d_1}{c_{P_1,k}\binom{|A|}{k}}
    \equiv_{p_1^{e_1}} \sum_{k=0}^{d_1}{c_{P_1,k}\binom{\tau}{k}} 
    =P_1(D_{\tau})
    \equiv_{p_1^{e_1}}1\\
P_2(A)&=\sum_{k=0}^{d_2}{c_{P_2,k}\binom{|A|}{k}}
    \equiv_{p_2^{e_2}} \sum_{k=0}^{d_2}{c_{P_2,k}\binom{0}{k}} 
    =P_2(\emptyset)
    \equiv_{p_2^{e_2}}0\\
\cdots \\    
P_r(A)&=\sum_{k=0}^{d_r}{c_{P_r,k}\binom{|A|}{k}}
    \equiv_{p_r^{e_r}} \sum_{k=0}^{d_r}{c_{P_r,k}\binom{0}{k}} 
    =P_r(\emptyset)
    \equiv_{p_r^{e_r}}0
\endaligned    
$$
where $D_{\tau}$ is any subset of $[n]$ with size $\tau$. 
This is a contradiction unless $|A|>n$. Hence we have 
$n<|A|\leq L=\Pi_{i=1}^r {L_i}\leq m\cdot\Pi_{i=1}^r{d_i}$.

\end{proof}

Consider $r=2$ and $e_1=e_2=1$ in above Theorem, we have $d_{p_1}(f)d_{p_2}(f)=\Omega(n)$. 
\begin{corollary}
Let $f\in \mathcal{F}_{n}^*$ be symmetric, then $f$ can has degree $o(\sqrt{n})$ in at most one finite field.
\end{corollary}

Note that for a fixed $I^c$ of size $|I^c|=k$, there are exactly $2^{n-k}$ partial assignments that leave the variables in $I^c$ free. Furthermore, for fixed $k$ there are ${\binom {n}{k}}2^{n-k}$ $k$-subcubes. We denote by $\mathcal{R}_{k,n}$ the set of restrictions that leave $k$ of $n$ variables free, i.e. $|\mathcal{R}_{k,n}|={\binom {n}{k}}2^{n-k}$. 
For a fixed $\rho\in \mathcal{R}_{k,n}$, we choose $f\in \mathcal{F}_n$ uniformly at random. 
Let $\mathbb{I}_{\rho}(f)$ be the indicator variable of the event that $f|_{\rho}$ is a non-junta symmetric subfunction on the corresponding $k$-subcube. 
Therefore, if $\mathbb{I}_{\rho}(f)=1$ then $f|_{\rho}\in \mathcal{F}_{k}^*$. Observe that $\Pr_{f\in \mathcal{F}_n}\left[\mathbb{I}_{\rho}(f)=1 \right]=(2^{k+1}-2)\big/{2^{2^k} }$.  

A natural question is how small can be $k$ to make sure the existence of a no-junta symmetric subfunction. We are going to apply the second moment method which in fact is a corollary of $\Pr[X=0]\le Var[X]/{E[X]^2}$.
\begin{lemma}(See \cite{AS08}.)\label{probmd}
Let $X=\sum_{i=1}^{M}X_i$, where $X_i$ is the indicator random variable for the event $A_i$. Denote $i\backsim j$ for the events $A_i$, $A_j$ that are not independent and define $\Delta=\sum_{i\backsim j}\Pr[A_i\wedge A_j]$. Thus $Var[X]\leq E[X]+\Delta$. If $E[X]\to \infty$ and $\Delta=o({E[X]}^2)$ then $X>0$ almost always. 
\end{lemma} 

Here we consider the above $X_i$ as $\mathbb{I}_{\rho_{i}}(f)$ and $A_i$ as the event that $f|_{\rho_i}$ is no-junta and symmetric. To acquire an upper bound of $\Delta$, we need the following fact.

\begin{fact}\label{dep}
For a fixed $k$-subcube $S$ in an $n$-cube, there are at most $\binom{n}{k}\cdot 2^k$ $k$-subcubes  intersecting with $S$. 
\end{fact}
\begin{proof}
For each vertex of the $n$-cube, it can only belong to at most $\binom{n}{k}$ different $k$-subcubes. A fixed $k$-subcubes contains $2^k$ vertices, so there are at most $\binom{n}{k}\cdot 2^k$ $k$-subcubes intersecting with the given $S$. 
\end{proof}

\begin{lemma} \label{aalgn}
For large $n$, if $k\leq \lg{n}-1$ then
almost always each $f\in \mathcal{F}_n$ has at least one $k$-variable symmetric subfunction.
\end{lemma}
\begin{proof}
Define $\mathbb{I}(f)=\sum_{\rho\in \mathcal{R}_{k,n}}{\mathbb{I}_{\rho}(f)}$.  
The goal is to show  $E_{f\in \mathcal{F}_n}\left[\mathbb{I}(f) \right]\to \infty$ and $\Delta=o({E_{f\in \mathcal{F}_n}\left[\mathbb{I}(f) \right]}^2)$. Then by Lemma \ref{probmd}, we have $\mathbb{I}(f)>0$ almost always. Since $\mathbb{I}(f)$ is the summation of indicators, this means almost always $\mathbb{I}(f)\geq 1$.

First observe that
$$
E_{f\in \mathcal{F}_n}\left[\mathbb{I}(f) \right] 
= {\binom {n}{k}}2^{n-k}\cdot\frac{ (2^{k+1}-2)}{2^{2^k}} $$
$$ > \frac{{\binom {n}{k}}2^{n-k}\cdot 2^{k}}{2^{2^k}}
> \frac{n^k 2^n}{k^k 2^{2^k}},
$$
where the last term is by $\binom {n}{k}>n^k/k^k$. It is easy to see that if $k\leq \lg{n}$ then  $2^{2^{k}}\leq 2^n$ and $n^k\gg k^k$. Therefore, as $n\to \infty$ and $k\leq \lg{n}$, $E_{f\in \mathcal{F}_n}\left[\mathbb{I}(f) \right]\to \infty$. 
Furthermore,
$$
\frac{1}{\left(E_{f\in \mathcal{F}_n}\left[\mathbb{I}(f) \right]\right)^2}
<\left( \frac{2^{2^k}}{ {\binom {n}{k}}\cdot 2^{n} } \right)^2 .
$$
On the other hand, by Fact \ref{dep}, 
$$
\Delta=\sum_{i\backsim j}\Pr[A_i\wedge A_j] 
< {\binom {n}{k}}2^{n-k} \cdot \binom{n}{k}2^k \cdot \frac{2^{k+1}}{2^{2^k}}.  
$$
Combining the above two inequalities,  we get
$$
\frac{\Delta}{\left(E_{f\in \mathcal{F}_n}\left[\mathbb{I}(f) \right]\right)^2}
<\frac{2^{2^k+k+1}}{2^n}
$$
Note that $k= \lg{n}-1\Rightarrow 2^k= n/2$. This makes ${\Delta}/{\left(E_{f\in \mathcal{F}_n}\left[\mathbb{I}(f) \right]\right)^2}\to 0$, i.e. $\Delta=o({\left(E_{f\in \mathcal{F}_n}\left[\mathbb{I}(f) \right]\right)^2})$. Hence by Lemma \ref{probmd}, if $n$ is large enough then $\mathbb{I}(f)\gg 1$ almost always.
\end{proof}


\begin{theorem}\label{thm_nonsymexact}
Let $f\in \mathcal{F}_{n}$ be a non-symmetric function.
Let $m=\Pi_{i=1}^{r}p_i^{e_i}$, where $p_i$'s are distinct primes, and $r$ and $e_i$'s are positive integers. 
Then almost always
$$
m\cdot d_{p_1^{e_1}}(f)\cdots d_{p_r^{e_r}}(f) > \lg{n}-1.
$$
\end{theorem}
\begin{proof}
For a given non-symmetric $f$, let $k=\lg{n}-1$, then the above lemma implies that almost always we can find a restriction $\rho\in \mathcal{R}_{k,n}$ for this $f$ such that $f|_{\rho}$ is no-junta and symmetric. Obviously $\forall i$ $d_{p_i^{e_i}}(f)\geq d_{p_i^{e_i}}(f|_{\rho})$.
Furthermore, $f|_{\rho}\in \mathcal{F}_{k}^*$. Hence, for $f|_{\rho}$ we have 
$m\Pi_{i=1}^{r}{d_{p_i^{e_i}}(f)} \geq m\Pi_{i=1}^{r}{d_{p_i^{e_i}}(f|_{\rho})} > \lg{n}-1$. 
\end{proof}

\begin{corollary}
With the same setting of the above theorem, and let $d_{max}=\max_{i\in [r]}(d_{p_i^{e_i}}(f))$ and $M=\max_{i\in [r]}(p_i^{e_i})$, then almost always
$$
d_{max}>\frac{\left(\lg{n}-1\right)^{1/r}}{M}.
$$
\end{corollary}

For the $\mathsf{OR}_n$ function, although in different representing models, this corollary is analogue to a result of \cite{BL14}, which states the degree lower bound of the weak representation by nonclassical polynomials is $\Omega{((\lg{n})^{1/r})}$. 
It is interesting to know why different proofs meet the same barrier of $\lg{n}$; does this relate to the entropy?

By taking $m=pq$ with distinct primes $p$ and $q$, we immediately have $d_p(f)d_q(f)=\Omega(\lg{n})$ almost always. 
\begin{corollary}
It is almost always true that $f\in \mathcal{F}_n$ has exact representing polynomial of degree $o(\sqrt{\lg{n}})$ over at most one finite field.
\end{corollary}

Note the lower bound in \cite{GLS10} means $f$ has degree $o(\lg{n})$ in at most one finite field. Although our result is weaker than \cite{GLS10}, our proof and the bound are both much simpler.

\section{Conclusions}
We prove any $n$-variate Boolean has no-junta subfunctions. It shows a possible way to deal with general Boolean functions via its subfunctions of under some restrictions. Besides,
for symmetric boolean functions, we proved $pqd_p(f)d_q(f)>n$, where $p$ and $q$ are distinct primes. This means any no-junta symmetric $f$ can have degree $o(\sqrt{n})$ in at most one finite field. In the nonsymmetric case, we prove that for a random function $f:\{0,1\}^n\to \{0,1\}$ almost always $pqd_p(f)d_q(f)>\lg(n)-1$, which means almost always $f$ can have degree $o(\sqrt{\lg{n}})$ in at most one finite field.



\begin{thebibliography}{1}

\bibitem{AS08} N. Alon and J. H. Spencer. {\it The Probabilistic Method}, 3ed. John Wiley and Sons, 2008. 

\bibitem{Am06} A. Ambainis. Polynomial degree vs. quantum query complexity, {\it Journal of Computer and System Sciences} Vol. 72(2), pp. 220-238, 2006.


\bibitem{BL14} A. Bhowmick and S. Lovett. Nonclassical polynomials as a barrier to polynomial lower bounds, {\it Electronic Colloquium on Computational Complexity} (ECCC), TR14-175, 2014.  


\bibitem{GLS10} P. Gopalan, S. Lovett, A. Shpilka. The  Complexity of Boolean Functions in Different Characteristics. {\it Computational Complexity} 19(2): pp.235-263. 2010.

\bibitem{Has86} J. H{\aa}stad. Computational limitations for small depth circuits, {\it Ph.D. thesis}, MIT Press, 1986. 

\bibitem{Jun12} S. Jukna, {\it Boolean Function Complexity: Advances and Frontiers}, Springer Publishing Company, Incorporated, 2012.

\bibitem{NS94} N. Nisan, M. Szegedy. On the degree of Boolean functions as real polynomials, {\it Computational Complexity}, 1994.



\bibitem{Raz87}  A. Razborov. Lower bounds for the size of circuits of bounded depth with basis $\{\wedge, \oplus\}$, {\it Mathematical Notes of the Academy of Science of the USSR}, 41:333-338, 1987.


\bibitem{Tal14} A. Tal. Shrinkage of De Morgan Formulae by Spectral Techniques, {\it 55th Annual Symposium on Foundations of Computer Science} (FOCS), 2014.

\bibitem{Tsai96} S. C. Tasi. Lower bounds on representing Boolean  functions as polynomials in $\mathbb Z_m$, {\it SIAM Journal on Discrete Mathematics}, vol. 9, No. 1, pp. 55-62, 1996.



\end{thebibliography}
\end{document}